\documentclass[12pt]{article}
\usepackage{graphicx}
\usepackage{amssymb}
\usepackage{amsthm}
\usepackage{gastex}
\usepackage{url}

\theoremstyle{plain}
\newtheorem{theorem}{Theorem}[section]
\newtheorem{corollary}[theorem]{Corollary}

\newtheorem{proposition}[theorem]{Proposition}

\theoremstyle{definition}
\newtheorem{definition}{Definition}[section]

\theoremstyle{remark}

\title{\Large{\textbf{Note on the Greedy Parsing Optimality for Dictionary-Based Text Compression}}}

\author{
Maxime Crochemore$^{1,2}$,
Alessio Langiu$^1$ \and
Filippo Mignosi$^3$
\\
\\
$^1$ King's College London, London, UK \\
\textit{\{Maxime.Crochemore, Alessio.Langiu\}@kcl.ac.uk} \and
$^2$ Universit\'e Paris-Est, Paris, France \and
$^3$ Universit\`a dell'Aquila, L'Aquila, Italy \\
\textit{Filippo.Mignosi@di.univaq.it} 
}

\makeatletter

\def\@maketitle{%
  \newpage
  \null
  \vskip 2em%
  \begin{center}%
  \let \footnote \thanks
    {\LARGE \@title \par}%
    \vskip 1.5em%
    {\large
      \lineskip .5em%
      \begin{tabular}[t]{c}%
        \@author
      \end{tabular}\par}%
  \end{center}%
  \par
  \vskip 1.5em}

\makeatother

\begin{document}

\maketitle

\begin{abstract}
Dynamic dictionary-based compression schemes are the most daily used data compression schemes since they appeared in the foundational papers of Ziv and Lempel in 1977, commonly referred to as LZ77.
Their work is the base of Deflate, gZip, WinZip, 7Zip and many others compression software. 
All of those compression schemes use variants of the greedy approach to parse the text into dictionary phrases.
Greedy parsing optimality was proved by Cohn et al. (1996) for fixed length code and unbounded dictionaries.
The optimality of the greedy parsing was never proved for bounded size dictionary which actually all of those schemes require.

We define the suffix-closed property for dynamic dictionaries and we 
show that any LZ77-based dictionary, including the bounded variants, 
satisfy this property. Under this condition we prove the optimality of the greedy parsing as a variant of the proof by Cohn et al.
\end{abstract}

\section*{Introduction}\label{sec:intro}

The foundational Ziv and Lempel LZ77 algorithm \cite{lz77} is the basis of almost all the famous dictionary compressors, like gZip, PkZip, WinZip and 7Zip.
They consider a portion of the previous text as a dictionary, i.e. they use a dynamic dictionary formed by the set of all the factors of the text up to the current position within a sliding window of fixed size.
A dictionary phrase refers to an occurrence of such phrase in the text by using the couple (length, offset), where the offset is the backward offset w.r.t. the current position. Since a phrase is usually repeated more than once along the text and since pointers with smaller offset have usually a smaller representation, the occurrence close to the current position is preferred.

Furthermore, in LZ77 based compression, the greedy approach is used to parse the text into phrases, i.e, in an iterative way, the longest match between the dictionary and the forwarding text is chosen.
This is commonly called the greedy phrase.
Some LZ77-based algorithms as Deflate algorithm and the compressors based on them, like gZip and PkZip, use variants of the greedy approach to parse the text. 
Deflate64 algorithm implemented in WinZip and 7zip, contains some heuristics to parse differently the text in order to improve the compression ratio, but its time complexity was never clearly stated. 

The research about dictionary-based data compression and parsing optimality 
produced in the last decades some noticeable results. Let us recall some of them within a brief historical overview.

In '73, the Wagner's paper (see \cite{DBLP:journals/cacm/Wagner73}) shows a $O(n\ |D|^2)$ dynamic programming solution for the parsing problem in the case of static dictionary, where 
$n$ is the text length, $D$ is the dictionary and $|D|$ is the dictionary cardinality, i.e. the number of phrases belonging to the dictionary. Dictionary phrases can overlap each other. 

In '74 Schuegraf et al. (see \cite{Schuegraf}) showed that the parsing problem is equal to the shortest path problem on a graph associated to both a text and a static dictionary.
Since that the full graph for a text of length $n$ can have $O(n^2)$ edges in the worst case and the minimal path algorithm has $O(V +E)$ complexity, we have another solution for the parsing problem of  $O(n^2)$ complexity.

In '76 Ziv and Lempel (see \cite{DBLP:journals/tit/LempelZ76}) introduced a new measure of complexity for a given text defined as the number of phrases produced by parsing the text with a dynamic prefix closed dictionary. This preliminary work early leads to the foundational dynamic dictionary-based compression methods presented in \cite{lz77,lz78}, a.k.a. LZ77 and LZ78, appeared in '77 and '78 respectively.
They both use an online greedy parsing that is simple and fast in practice. 
Those compression methods use both an uniform (constant) cost model for the dictionary pointers, i.e. they use bounded size dictionaries and fixed length code for dictionary phrase references. The greedy approach used to parse the text is realized by choosing the longest match between the dictionary phrases and the forwarding text, scanning the text left to right, until the whole text is covered. After any dictionary phrase in the parsing, that can also be the empty word, a single plain text symbol is used. This guaranteed the existence of a parsing for any text and any dictionary.

In '82, the LZSS compression algorithm, based on the LZ77 one, was presented (see \cite{DBLP:journals/jacm/StorerS82}). It improves the compression ratio and the execution time without changing the original parsing approach. The main difference is that a symbol is used only when there is no match between dictionary and text. 
It uses a flag bit to distinguish symbols from dictionary pointers in the parsing. In the same paper Storer et al. proved the optimality of the greedy parsing for the original LZ77 scheme with unbounded dictionary (see the Theorem 10 in \cite{DBLP:journals/jacm/StorerS82} with $p=1$).

In '84, LZW variant of LZ78 was introduced by Welch (see \cite{lzw}). This is one of the firsts theoretical compression method that use a dynamic dictionary and variable costs of pointers. The main difference w.r.t. LZ78 is that the text is supposed to be composed by symbol from a fixed alphabet, knew in advance. The dictionary is initialized with all the alphabet symbols. This guaranteed that there will be always at least one dictionary phrase matching a factor of the text starting at any position. Exploiting this property, the parsing is composed just by dictionary phrases, without using explicit symbols, leading to a better compression. The LZW scheme has been very appreciated by the research community, indeed plenty of LZW variants have been presented so far.

In '85, Hartman and Rodeh proved in \cite{rodeh1985} the optimality of the 
\emph{one-step-lookahead parsing} for prefix-closed static dictionary and uniform pointer cost. The main point of this approach is to chose the phrase that is the first phrase of the longest match between \emph{two} dictionary phrases and the text. In other words, if the current parsing cover the text up to the $i$th character, then it choose the phrase $w$ such that $ww'$ is the longest match with the text starting at position $i$, with $w,w'$ belonging to the dictionary.

In '89 and later in '92, the \emph{deflate} algorithm was presented and used in PkZip and gZip compressors. It uses a LZ77-like dictionary, the LZSS flag bit and variants of the greedy parsing. Both dictionary pointers and symbols are encoded by using a Huffman code. Those compression schemes early become so popular to be included in many communication protocol, commercial compression software and transmission devices. 

In '95, Horspool investigated in \cite{DBLP:conf/dcc/Horspool95} about the effect of non-greedy parsing in LZ-based compression. He highlighted that using the above \emph{one-step-lookahead parsing} in the case of dynamic dictionaries leads to better compression w.r.t. the one obtained by using the greedy parsing. Horspool showed some experimental results using the LZW algorithm and a new LZW variant that he presented in the same paper. 

In '96 the greedy parsing was ultimately proved by Cohn et al. (see \cite{DBLP:conf/dcc/CohnK96}) to be optimal for static suffix-closed dictionary under the uniform cost model. They also proved that the right to left greedy parsing is optimal for prefix-closed dictionaries. Notice that the greedy parsing can be computed in linear time. Since the LZ77 dictionary is ``assumed" to be suffix-closed, this is a more general result w.r.t. the previous Storer et al. one. We present more details about LZ77 dictionary and the suffix-closed property in the next section.

In '99, Matias and Sahinalp (see \cite{DBLP:conf/soda/MatiasS99}) gave a linear-time optimal parsing algorithm in the case of prefix-closed dynamic dictionary and uniform cost of dictionary pointer, i.e. the codeword of all the pointers have equal length. They extended the results given in \cite{rodeh1985}, \cite{DBLP:conf/dcc/Horspool95} and \cite{cancan} to the dynamic case. 
Matias and  Sahinalp called their parsing algorithm
\emph{Flexible Parsing}. It is also known as semi-greedy parsing.

In '09, Ferragina et al. (see \cite{FerraginaSODA09}) introduced an optimal parsing algorithm for LZ77-like dictionary and variable length code, where the code length is assumed to be the cost of a dictionary pointer. In this paper the parsing optimality refers to the compression optimality, i.e. the parsing which leads to the better compression.

In '10, Crochemore et al. (see \cite{cglmr_iwoca10} and the extended version \cite{cglmr_JDA2011}) introduced an optimal parsing for prefix-closed dictionaries and variable pointer costs. It was called  \textit{dictionary-symbolwise flexible parsing} and it fits to both the LZ77 and the LZ78 dictionary cases.
It uses a graph-based model for the parsing problem where each node represent a position in the text and edges represent dictionary phrases. Edges are weighted according to the bit length of the encoded length and offset pair. It works for the original LZ77 and LZ78 algorithms and for almost all of their known variants. Recently, a new data structure called Multilayer Suffix Tree was presented (see \cite{dccpaper2013}) to address the a weak version of the rightmost position problem, strictly related with the parsing optimality problem.

\medskip

The main goal of this paper is to better explain the relationship between the LZ77 dictionary variants and the suffix-closed property and to prove the optimality of the greedy parsing for all of those cases. 
This paper is organized as follow. In Section \ref{sec:suffixclosed} we formally define the suffix-closed property for dynamic dictionaries and we  show that any LZ77-based dictionary, including the bounded variants, satisfy this property. 
In Section \ref{sec:optimality} we prove the optimality of the greedy parsing for suffix-closed dictionaries as a variant of the proof by Cohn et al.

\section{Suffix-Closed Dynamic Dictionaries}\label{sec:suffixclosed}

In data compression field, a dictionary is a set of finite length sequences or phrases. It is shared between compressor and decompressor. A static dictionary is a fixed set of phrases that does not change along the compression-decompression process. It is known in advance w.r.t. to the input text. The weakness of this model is that the dictionary does not depend by the text and, therefore, it cannot get adapted to it. This leads to poor compression results for those text having few overlap with the used dictionary.

A dynamic dictionary is a set of phrases that can change along the compression-decompression process. It can be the empty set at the very beginning of the compression process or it can be already initialized. Subsequently, it get populated accordingly to a dictionary algorithm. Usually, also phrase deletion are supported in order to limit the dictionary size. Given a text $T$ of length $n$, for any point in time $0 \leq i < n$, we call $D_i$ the dictionary 
at time $i$ of the compression or decompression process, i.e. $D_i$ is the dictionary after that the first $i$ symbols of the text have already been processed.

A static dictionary $D$ is prefix-closed (suffix-closed) if and only if for any phrase $w\in D$ in the dictionary, all the prefixes (suffixes) of $w$ belong to the dictionary, i.e. $\mbox{\emph{suff}}(w)\subset D$ ($\mbox{\emph{pref}}(w)\subset D$).  
For instance, the dictionary $D=\{a, ba, aba, bba\}$ is suffix-closed.

The LZ77 dictionary is defined as the set of factors of a portion of the already processed text. In other world, for any text $T$ and at any time $i$, the dictionary is the set of factors of the text fitting to a sliding window of length $h$, i.e.  \emph{fact}$(T[i-h+1..i])$. At any time $i$, the dictionary $D_i$ is both prefix- and suffix-closed. 
The LZ78 dictionary is maintained, starting from the empty set, by inserting a phrase formed by a symbol concatenated to the greedy phrase. For instance, if at the moment $i$ the greedy phrase matching the text is $T[i..j]$, $i \leq j$, then the next dictionaries $D_{k}$, $i < k \leq +1j$, are set to $D_{i} \cup T[i..j]T[j+1]$. This construction algorithm maintains the prefix-closed property for any $D_i$ dictionaries.

The classic Cohn' and Khazan's result of '96 (see \cite{DBLP:conf/dcc/CohnK96}) states that if $D$ is a static suffix-closed dictionary, then the greedy parsing is optimal under the uniform cost assumption. 
Symmetrically, the reverse of the greedy parsing on the reversed text is optimal for static prefix-closed dictionary. Roughly speaking, the original proof concerns with suffix-closed dictionaries and shows that choosing the longest dictionary phrases guarantees to cover the text with the minimum number of phrases. 
Unfortunately, sice LZ77 and LZ78 dictionaries are not static, above results does not apply to them.

Let us focus on the suffix-closed definition of dynamic dictionaries.
Let us recall that a static dictionary is a set of words $D$.
A static dictionary is suffix-closed if and only if for any factor $w$ in the dictionary $D$ the set of suffixes $\mbox{\emph{suff}}(w)$ of $w$ is a subset of the dictionary, i.e. $\mbox{\emph{suff}}(w)\subset D$.
Turning into the dynamic settings, let us say that 
at any moment $i$, $0\geq i$,
a dynamic dictionary $D_i$ is a set of words.
The suffix-closed and the prefix-closed property have been commonly considered \emph{naturally} extended, without a formal definition, to the dynamic case with the additional condition ``at any time".
Therefore, what is commonly meant as a suffix-closed dynamic dictionary is just that, at any time $i$, the dictionary $D_i$ has the suffix-closed property.

Notice that this definition does not make any assumption on the relationship between dictionaries at two different moments and it does not suffice to extend the parsing optimality for static dictionary to the dynamic case.

We define the suffix-closed property for dynamic dictionaries as follows.

\begin{definition} \label{prp:dynsufclosed}
A dynamic dictionary $D$ has the suffix-closed property \textit{iff}, at any moment $i$,
for any dictionary phrase $w\in D_i$ and 
for any $0 \leq k < |w|$, 
the suffix $w_k=w[k..|w|-1]$ of $w$ of length $|w|-k$ 
is in $D_{i}$ \emph{and} in $D_{i+k}$.
\end{definition}

Notice that the above suffix-closed property imply the \emph{natural} one.

We say that a dictionary is \textit{non-decreasing} when $D_i\subset D_j$ for any $i,j$ points in time, with $i \leq j$. A static dictionary is obviously non-decreasing. Practically speaking, a dynamic dictionary is non-decreasing when it can only grow along the time. 
For instance, the original LZ78 dictionary is a non-decreasing dictionary because, at each algorithm step, one phrase is inserted into the dictionary. On the contrary, 
many practical implementation and variants of LZ78 dictionary are not non-decreasing. Because of space saving purpose, the size of the dictionary is bounded in practice and some phrases are deleted from the previou dictionary.

Since the LZ77 dictionary is defined as the set of factors of a sliding window, i.e. the backward text up to a certain distance,  LZ77 has not the non-decreasing property.

\begin{proposition}\label{pro:lz77strong}
The original LZ77 bounded dictionary is suffix-closed. 
The unbounded variant of the LZ77 dictionary is non-decreasing and suffix-closed.
\end{proposition}

\begin{proof}
Recall that, for any text $T$, the LZ77 dictionary is equal to  \emph{fact}$(T[i-h+1..i])$.
The LZ77 dictionary is unbounded when $h\geq |T|$. In this case, for any $i<|T|$, the dictionary $D_i$ is equal to the set $\mbox{\emph{fact}}(T[0 : i])$ that is obviously non-decreasing.

Let us focus on the general set $\mbox{\emph{fact}}(T[i-h+1 : i])$. 
This is a sliding window of size $h$ over the text $T$.
By the very definition of the set of factors, the dictionary $D_i$ is suffix-closed, at any moment $i$.
For any value $i$, let be $T[i-h+1 : i]=au$ and $T[i-h+2 : i+1]=ub$ with $a,b$ in $\Sigma$ and $u$ in $\Sigma^*$. Since all the proper suffixes of $au$ are also suffixes of $u$, then for any $w\in \mbox{\emph{fact}}(T[i-h+1 : i])=D_i$ the proper suffixes $w_k$ of length $|w|-k$, $1\leq k < |w|$, are also in $\mbox{\emph{fact}}(T[i-h+2 : i+1])=D_{i+1}$. 
Therefore, any proper suffix of $w\in D_i$ is also in $D_{i+1}$, for any $i$, $w$.
It easy to see that this property is equivalent to the suffix-closed property defined in Definition \ref{prp:dynsufclosed}.
\end{proof}

Let us now to refer to the effect of the prefix- and suffix-closed properties on the graph-based model of the parsing problem in order to visualize those concepts. 
Given a text $T$ and a
dictionary $D$, if $D$ has the strong suffix-closed property, then for any edge $(i,j)$ of the graph $G_{D,T}$ associated with the phrase $w\in D_i$, with $|w|=j-i$ and $w=T[i:j]$, then all the edges $(k,j)$, $i<k<j$ are into $G_{D,T}$. 
In the case of prefix closed dictionaries, as prefix edges start from the same node, the prefix of a dictionary phrase are all represented in the graph if the dictionary has just the \emph{natural} prefix-closed property.

\section{Greedy Parsing Optimality}\label{sec:optimality}

We want now to extend 
the elegant proof of Cohn et al. (see \cite{DBLP:conf/dcc/CohnK96}) to the case of suffix-closed dynamic dictionaries.

Given a text $T$ of length $n$ and a dynamic dictionary $D$ where, at the moment $i$-th with $0\leq i < n$, the text $T_i$ has been processed and $D_i$ is the dictionary  at time $i$. Recall that we are under the uniform cost assumption.

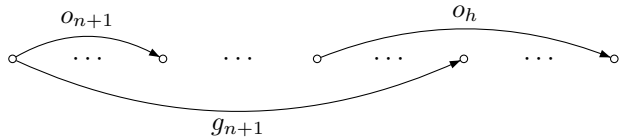
\begin{figure}[t] \centering{\footnotesize{\begin{picture}(120,20)(0,-10)
\gasset{Nw=1,Nh=1}
   \node(2)(20,0){}
   \node[Nframe=n](L)(30,0){$\cdots$} 
   \node(3)(40,0){}
   \node[Nframe=n](L)(50,0){$\cdots$} 
   \node(4)(60.5,0){}
   \node[Nframe=n](L)(70,0){$\cdots$} 
   \node(5)(80,0){}
   \node[Nframe=n](L)(90,0){$\cdots$} 
   \node(6)(100,0){}
   
   \drawedge[curvedepth=3](2,3){$o_{n+1}$}
   \drawedge[ELside=r,curvedepth=-7](2,5){$g_{n+1}$}
   \drawedge[curvedepth=4](4,6){$o_{h}$}
\end{picture}}
 \vspace{0.30cm}
\caption{Detail of the differences between parsing $\mathcal{O}$ and parsing $\mathcal{G}$ over a text $T$ between positions $|o_1\cdots o_n|$ and $|o_1\cdots o_h|$. Nodes and dots represent the text and edges represent parsing phrases as reported on edge labels.}
\label{fig:cohn}
}\end{figure}

\begin{theorem}\label{the:cohngeneralizzato}
The greedy parsing of $T$ is optimal for \textit{strong} suffix-closed dynamic dictionaries.
\end{theorem}

\begin{proof}
The prove is by induction. We want to prove that for any $n$ smaller than or equal to the number of phrases of an optimal parsing, there exists an optimal parsing where the first $n$ phrases are greedy phrases.
The inductive hypothesis is that there exists an optimal parsing where the first $n-1$ phrases are greedy phrases. We will prove that there is an optimal parsing where the first $n$ phrases are greedy and, therefore, any greedy parsing is optimal.
We use here the notation $w^k$ to refer to the suffix of $w$ of length $|w|-k$.

Fixed a text $T$ and a strong suffix-closed dynamic dictionary $D$, let $\mathcal{O}=o_1o_2\cdots o_p=T$ be an optimal parsing and let $\mathcal{G}=g_1g_2\cdots g_q=T$ be the greedy parsing, where, obviously, $p\leq q$. 

The base of the induction with $n=0$ is obviously true. Let us prove the inductive step.

By inductive hypothesis,  $\forall\ i<n$ we have that $o_i=g_i$.
Since $g_{n}$ is greedy, then the $n$-th phrase of the greedy parsing is longer than or equal to the $n$-th phrase of the optimal parsing, i.e. $|g_{n}|\geq |o_{n}|$ and therefore $|o_1\cdots o_{n}|\leq |g_1\cdots g_{n}|$.

If $|g_{n}| = |o_{n}|$, then the thesis follows.
Otherwise, $|g_{n}| > |o_{n}|$ and $o_n$ is the first phrase in the optimal parsing that is not equal the $n$-th greedy phrase.

Let $h$ be the minimum number of optimal parsing phrases that overpass $g_n$ over the text, i.e. $h=min\{i\ |\ |o_1\cdots o_{i}| \geq |g_1\cdots g_{n}|\}$. Since $|g_{n}| > |o_{n}|$, then $h>n$.
If $|o_1\cdots o_{h}| = |g_1\cdots g_{n}|$, 
then the parsing $g_1 \cdots g_n o_{h+1} \cdots o_p$ 
uses a number of phrases strictly smaller than the number of phrases used by the optimal parsing that is a contradiction.
Therefore $|o_1\cdots o_{h}| > |g_1\cdots g_{n}|$. The reader can see this case reported in Figure \ref{fig:cohn}.

Let $|o_1\cdots o_{h-1}|=T_j$ be the text up to the $j$-th symbol. Then $o_h\in D_j$, where $D_j$ is the dynamic dictionary at the time $j$. 
Let $o_{h}^k$ the $k$-th suffix of $o_h$ with $k=|o_1\cdots o_{h}| - |g_1\cdots g_{n}|$. 
For the Property \ref{prp:dynsufclosed} of $D$, $o_{h}^k \in D_{j+k}$ and then there exists a parsing $o_1 \cdots o_{n-1} g_n o_{h}^k o_{h+1} \cdots  o_p$, where $g_n o_{h}^k=o_n\cdots o_h$.

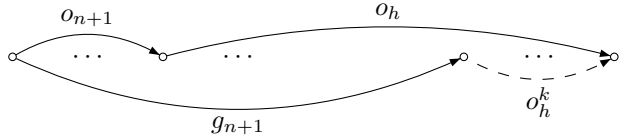
\begin{figure}[t] \centering{\footnotesize{\begin{picture}(120,20)(0,-10)
\gasset{Nw=1,Nh=1}
   \node(2)(20,0){}
   \node[Nframe=n](L)(30,0){$\cdots$} 
   \node(3)(40,0){}
   \node[Nframe=n](L)(50,0){$\cdots$} 
   \node(5)(80,0){}
   \node[Nframe=n](L)(90,0){$\cdots$} 
   \node(6)(100,0){}
   
   \drawedge[curvedepth=3](2,3){$o_{n+1}$}
   \drawedge[ELside=r,curvedepth=-7](2,5){$g_{n+1}$}
   \drawedge[curvedepth=4](3,6){$o_{h}$}
   \drawedge[ELside=r,dash={1.5}{1.5},curvedepth=-3](5,6){$o_{h}^k$}
\end{picture}}
 \vspace{0.30cm}
\caption{Detail of the differences between parsing $\mathcal{O}$ and parsing $\mathcal{G}$ over a text $T$ between positions $|o_1\cdots o_n|$ and $|o_1\cdots o_h|$. Nodes and dots represent the text and edges represent parsing phrases as reported on edge labels. The dashed edge $o_{h}^k$ represents a suffix of $o_{h}$.}
\label{fig:cohn2}
}\end{figure}

From the optimality of $\mathcal{O}$, it follows that $h=n+1$, otherwise there exists a parsing with less phrases than an optimal one. See Figure \ref{fig:cohn2}.
Therefore $o_1 \cdots o_{n-1} g_n o_{n+1}^k o_{n+2} \cdots  o_p$  is also  an optimal parsing. 
Since $o_1 \cdots o_{n-1}$ is equal to $g_1 \cdots g_{n-1}$, the thesis follows.
\end{proof}

\begin{corollary}
The greedy parsing is an optimal parsing for any version of the LZ77 dictionary.
\end{corollary}
The proof of the above corollary comes straightforward from the Theorem  \ref{the:cohngeneralizzato} and the Proposition \ref{pro:lz77strong}.

\medskip

To our best knowledge, this is the first proof of optimality of the greedy parsing that cover the original LZ77 dictionary case and almost all of the practical LZ77 dictionary implementations where the search buffer is a sliding windows on the text.

\bibliographystyle{abbrv}
\bibliography{compression}

\end{document}